\documentclass[preprint,report]{elsarticle}

\usepackage{array,xspace,multirow,hhline,tikz,colortbl,tabularx,booktabs,fixltx2e,amsmath,amssymb,amsfonts,amsthm}
\usepackage{paralist}
\usepackage[boxed]{algorithm}
\usepackage[noend]{algorithmic}
\usepackage{amssymb,amsthm,xspace}

\usepackage{verbatim,ifthen}
\usepackage{enumitem}
\usepackage{url}
\usepackage{pifont}
\usepackage{ifthen}
\usepackage{calrsfs,mathrsfs}
\usepackage{bbding,pifont}
\usepackage{pgflibraryshapes}
\usetikzlibrary{arrows}
\usepackage{centernot}
\usetikzlibrary{decorations.pathreplacing}
\usetikzlibrary{patterns}
\usepackage{pgflibraryshapes}
\usetikzlibrary{positioning,chains,fit,shapes,calc}
	\usepackage{varioref}
	\usepackage{nicefrac}

	\usepackage{tabularx}
	\usepackage{mathtools}

	\algsetup{linenodelimiter=\,}
	\algsetup{linenosize=\tiny}
	\algsetup{indent=2em}

\definecolor{light-gray}{gray}{0.9}

\newtheorem{definition}{Definition}%

\usepackage{boxedminipage}
\usepackage{xspace}

\newcommand{\pbDef}[3]{%
\noindent
\begin{center}
\begin{boxedminipage}{0.98 \columnwidth}
#1\\[5pt]
\begin{tabular}{l p{0.75 \columnwidth}}
Input: & #2\\
Question: & #3
\end{tabular}
\end{boxedminipage}
\end{center}
}

		\newcommand{\jr}{\ensuremath{\mathit{JR}}\xspace}

				\newcommand{\bt}[1][]{\ensuremath{\ifthenelse{\equal{#1}{}}{\mathit{BT}}{\mathit{BT}(#1)}}\xspace}

			\newcommand{\pref}{\ensuremath{\succsim}}

	\newtheorem{lemma}{Lemma}%
	\newtheorem{proposition}{Proposition}%
	\newtheorem{example}{Example}

\usepackage{eqparbox}

	\usepackage{enumitem}
	\setenumerate[1]{label=\rm(\it{\roman{*}}\rm),ref=({\it\roman{*}}),leftmargin=*}
	\newlength{\wordlength}

\usepackage{enumitem}
\setenumerate[1]{label=\rm(\it{\roman{*}}\rm),ref=({\it\roman{*}}),leftmargin=*}

\newcommand{\nbh}[1][]{
	\ifthenelse{\equal{#1}{}}{\nu}{\nu(#1)}
}

\newcommand{\cstr}[1][]{
	\ifthenelse{\equal{#1}{}}{\mathscr S}{\cstr(#1)}
}

\newcommand{\choice}[1][]{
	\ifthenelse{\equal{#1}{}}{\mathit{C}}{\choice(#1)}
}

\begin{document}

	 
	  
		  
		  \title{Sub-committee Approval Voting \\and Generalised Justified Representation Axioms}

	\author{Haris Aziz and Barton Lee} \ead{haris.aziz@data61.csiro.au}
	
	\address{Data61, CSIRO and UNSW,\\ Sydney, Australia}
	
%

\begin{abstract}
Social choice is replete with various settings including single-winner voting, multi-winner voting, probabilistic voting, multiple referenda, and public decision making. We study a general model of social choice called Sub-Committee Voting (SCV) that simultaneously generalizes these settings. 
We then focus on sub-committee voting with approvals and propose extensions of the justified representation axioms that have been considered for proportional representation in approval-based committee voting. We study the properties and relations of these axioms. For each of the axioms, we analyse  whether a representative committee exists and also examine the complexity of computing and verifying such a committee. 
\end{abstract}

\maketitle

\section{Introduction}

Social choice is a general framework of preference aggregation in which voters express preferences over outcomes and a desirable outcome is selected based on the preferences of the voters~\citep{ABES17a,Coni10a}. The most classic model of social choice is \emph{(single winner) voting} in which voters express preferences over a set of alternatives and exactly one alternative is selected~\citep{BrFi02a}. A natural generalization of the model is \emph{muti-winner voting} or \emph{committee voting} in which a set of alternatives is selected~\citep{FSST17a}. Another model is \emph{multiple referenda} in which voters vote over multiple but independent binary decisions~\citep{BKZ97a}. Probabilistic versions of single-winner voting have also been examined~\citep{Gibb77a}. 

In this paper, we study a natural model of social choice that simultaneously generalizes all the social choice settings mentioned above. The advantage of considering a more general combinatorial model~\citep{LaXi15a} is that instead of coming up with desirable axioms, rules, and algorithms in a piecemeal manner for different settings, one can design or apply general principles and approaches that may be compelling for a wide range of settings. Of course certain axioms may only be meaningful for a certain subsetting but as we show in this paper, a positive algorithmic or axiomatic result for well-justified axioms can be viewed favourably for all relevant sub-settings as well.
 Another advantage of formalising a general model is that it provides an opportunity to unify different strands of work in social choice. 
Our model also helps approach the committee voting problem in which there are additional diversity constraints possibly relating to gender, race or skill.  
Finally, our model applies to general participatory budgeting scenarios~\cite{Caba04a} where multiple decisions needs to be made and the minority representation needs to be protected.

After formalizing the SCV (sub-committee voting) setting, we focus on a particular restriction of SCV in which agents or voters only express approvals over some of the alternatives or candidates. 
The restriction to approvals is desirable because approvals capture dichotomous/binary preferences that are prevalent in many natural settings. Secondly, ordinal and cardinal preferences coincide when preferences are dichotomous. This is desirable since elicitation of cardinal utilities has been considered controversial in decisions concerning public goods. 

SCV with approvals can be viewed as a multidimensional generalization of approval-based committee voting~\citep{Kilg10a}.
For approval-based committee voting, a particularly appealing axiom that captures representation is justified representation (\jr) that requires that a set of voters that is large enough and cohesive enough in their preferences should get at least one approved candidate in the selected committee. The axiom has received considerable attention~\citep{BFJL16a,SFF+17a,SFF16a}. 
For SCV with approvals, we extend the justified representation axiom~\citep{ABC+16a} that has only been studied in the context of committee voting.

 One interesting application captured by this SCV framework, which is not possible under standard models, is committee voting in the presence of diversity constraints or quotas. Considering this application highlights the conflict between diversity constraints and the original axioms of fair, or justified, representation. As will be shown this conflict leads to conceptual issues of what is the `appropriate' generalisation of the \jr axiom for SCV instances and also technical issues such as existence and computational intractability of achieving certain axioms whilst diversity constraints are enforced. 

\paragraph{Contributions}

Our contributions are threefold with the first two being conceptual contributions. 

Firstly, we study a natural model of social choice called SCV (sub-committee voting) that simultaneously generalizes several previously studied settings.

Secondly, we focus on approval-based SCV and present new notions of justified representation (\jr) concepts including Intra-wise JR (IW-JR) and Span-wise JR (SW-JR). These distinct notions lead to `local' and `global' approaches to representation, respectively. 

Thirdly, we present technical results concerning the extent to which these properties can be satisfied. We show that although SW-JR is a natural extension of JR to the SCV setting, a committee satisfying SW-JR may not exist even under severe restrictions. Furthermore, checking whether there exists a committee satisfying SW-JR is NP-complete. The results always show that the more general setting SCV is considerably more challenging than approval-based committee voting. 
We then formalize a weakening of SW-JR called weak-SW-JR and 
present a polynomial-time algorithm that finds a committee that \emph{simultaneously} satisfies weak-SW-JR and IW-JR.  
We also propose two natural generalizations of PAV (Proportional Approval Voting), a well-known rule for committee voting under approvals. However we show that neither of these two extensions satisfies both weak-SW-JR and IW-JR.
 


%

\section{Sub-committee Voting}

We propose a new setting called SCV that generalizes a number of voting models. The setting is a tuple $(N,C,\pi,q,\pref)$

\begin{itemize}
	\item $N=\{1,\ldots n\}$ is the set of voters/agents. 
	\item $C=\{c_1,\ldots, c_m\}$ is the set of candidates.
	\item $\pi=\{C_1,\ldots, C_{\ell}\}$ is a partitioning of the candidates. Each $C_j$ is referred to
	as a \emph{candidate subset} from which a \emph{sub-committee} is to be chosen.
	\item $q$ is the quota function that specifies the number of candidates $q(C_j)=k_j$ to be selected from each subset $C_j$. 
We denote $\sum_{j=1}^\ell k_j$ by $k$. 
\item ${{\pref}=(\pref_1,\ldots, \pref_n)}$ specifies for each agent $i$, her preferences/utilities over $C$. We allow the possibility that an agent does not compare candidates across candidate subsets. At a minimum it is required that each $\pref_i$ is transitive and complete within each subset $C_j$, however additional restrictions can be introduced, as befitting the setting; they might even be replaced with cardinal utilities
\end{itemize}

An SCV outcome $p$ specifies a real number $p(c)$ for each $c\in C$ with the following constraints:
\begin{align*}
	0\leq p(c)\leq 1&\quad \text{for all } c\in C\\
	\sum_{c\in C}p(c)=k& \quad \text{and  }
	\sum_{c\in C_j}p(c)=k_j \text{ for all } j\in \{1,\ldots, \ell\}.
	\end{align*}
	
In this paper we restrict our attention to discrete outcomes so that $p(c)\in \{0,1\}$ but in general SCV can allow for probabilistic outcomes where $p(c)$ is the probability of selecting candidate $c$.
For discrete outcomes, an outcome $W$ will be a committee of size $k$ that consists of $\ell$ sub-committees $W_1,\ldots, W_{\ell}$ where each $W_j\subseteq C_j$ and $|W_j|=k_j$.

	 If $\ell=1$, $p(c)\in \{0,1\}$ for all $c\in C$ and $k=1$, we are in the \emph{voting setting}. 
	 If $\ell=1$, $p(c)\in \{0,1\}$ for all  $c\in C$, we are in the \emph{committee/multi-winner voting setting}~\citep{FSST17a}. 
	 If $\ell=1$ and $k=1$, we are in the \emph{probabilistic voting setting}~\citep{Gibb77a}. 
	 If $p(c)\in \{0,1\}$ for all $c\in C$ and $k_i=1$ for all $i\in \{1,\ldots, \ell\}$, we are in the \emph{public decision making setting}~\citep{CFS17a}. Note that public decision making setting is equivalent to the ``voting on combinatorial domain'' setting studied by \citet{LaXi15a}. The latter setting allows for more complex preferences over the set of combinatorial outcomes but the preferences may not be polynomial in the number of candidates and voters. 
	 If $p(c)\in \{0,1\}$ for all $c\in C$ and $k_i=1$ for all $i\in \{1,\ldots, \ell\}$ and $|C_j|=2$ for all $j\in \{1, \ldots, \ell\}$, we are in the \emph{multiple-referenda setting}~\citep{BKZ97a,LaNi00a,CuLa12a}. 	
	

				\begin{figure}[h!]
					
					\begin{center}
						\scalebox{0.9}{
			\begin{tikzpicture}
				\tikzstyle{pfeil}=[->,>=angle 60, shorten >=1pt,draw]
				\tikzstyle{onlytext}=[]

\node        (SCV) at (-5,6) {\begin{tabular}{c} Sub-committee\\ Voting \end{tabular}};;

	 \node[onlytext] (PDM) at (-5,4) {\begin{tabular}{c} Public\\ Decision\\Making \end{tabular}};
	  \node[onlytext] (MR) at (-5,2) {\begin{tabular}{c} Multiple\\ Referenda \end{tabular}}; 
	 \node[onlytext] (RA) at (-8,2) {\begin{tabular}{c} Resource \\ Allocation \end{tabular}};
	  \node[onlytext] (MWV) at (-1,4) {\begin{tabular}{c} Multi\\ Winner\\Voting \end{tabular}};
	   \node[onlytext] (SW-V) at (-1,2) {\begin{tabular}{c} Single\\ Winner\\Voting \end{tabular}};

  			\draw[<-] (SCV) -- (PDM) ;
			\draw[<-] (SCV) -- (MWV) ;
			\draw[<-] (MWV) -- (SW-V) ;
			\draw[<-] (PDM) -- (SW-V) ;
			\draw[<-] (PDM) -- (RA) ;
				\draw[<-] (PDM) -- (MR) ;
			\end{tikzpicture}
			}
			\end{center}
			
			\caption{Relations between settings. 
			 An arrow from (A) to (B) denotes that setting (A) is a restriction of setting (B).}
			\end{figure}
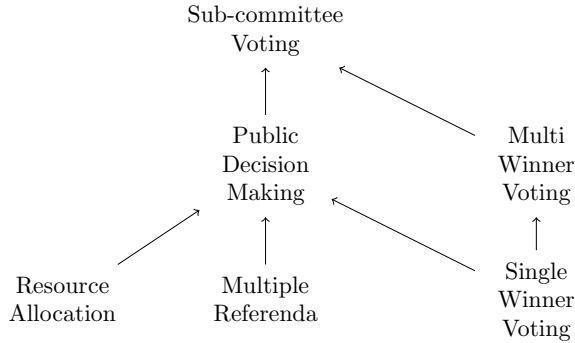

At a very abstract level, even a bicameral legislature can be viewed as an SCV setting in which $\ell=2$. Even if there are no explicit multiple committees, there can be diversity constraints imposed on the committee that can be easily modelled as an SCV problem. For example, the problem of selecting five people with 3 women and 2 men can be viewed as an SCV problem with two candidate subsets.

In this paper, we will focus exclusively on approval-based voting in the SCV setting. In approval-based voting we replace each agent $i$'s  preference $\pref_i$ with an approval ballot $A_i\subseteq C$ which represents the subset of candidates that she approves.  The list $A = (A_1,\ldots, A_n)$ of approval ballots is referred to as the {\em ballot profile}. As per the general SCV setting introduced at the start of the section, the goal is to select a target $k$ number of candidates from $C$ which satisfy the quota function for each candidate subset.

\section{Justified Representation in Approval-based Sub-committee Voting}

%

We now focus on the SCV setting in which each agent approves a subset of the candidates. Based on the approvals, the goal is to identify a fair or representative outcome. Note that if $\ell=1$, we are back in the committee voting setting. The approval-based SCV setting can be seen as capturing $\ell$ independent committee voting settings. 

For the approval-based committee voting setting, justified representation (\jr) is a desirable property.

\begin{definition}[Justified representation (JR)]
Given a ballot profile $A = (A_1, \dots, A_n)$ over a candidate set $C$ and a target committee size $k$,
we say that a set of candidates $W$ of size $|W|=k$ {\em satisfies justified representation 
for $(A, k)$} if $\forall X\subseteq N$:
\[ |X|\geq \frac{n}{k} \text{ and } |\cap_{i\in X}A_i|\geq 1 \implies (|W\cap (\cup_{i\in X}A_i)|\geq 1)\]

\end{definition}

One natural extension of \jr to the case of SCV is to treat each candidate subset as an independent committee voting problem. Then an SCV outcome satisfies 
Intra-wise JR (IW-JR) if each sub-committee satisfies \jr.

\begin{definition}[Intra-wise JR (IW-JR)]
	An SCV outcome $W$ satisfies \emph{Intra-wise JR (IW-JR)} if $\forall X\subseteq N$ and $\forall j\in \{1,\ldots,\ell\}$: 
	\begin{align*}
	& |X|\geq \frac{n}{k_j} \text{ and } |(\cap_{i\in X}A_i)\cap C_j|\geq 1 \\
	&\implies (|W\cap  C_j\cap (\cup_{i\in X}A_i)|\geq 1)
	\end{align*}
	\end{definition}

	We note that since a committee satisfying \jr can always be attained by a polynomial-time algorithm~\citep{ABC+16a}, IW-JR is easy to achieve by treating each sub-committee voting as a separate committee voting problem. 
	
	The limitation of this approach is that it could be that each time the same voters are unrepresented in each sub-committee and they may ask for some representation in at least some sub-committee. Thus IW-JR can be considered as a `local' \jr axiom which ignores whether or not a given voter has already been represented in some other sub-committee.
	
	In view of this limitation, another extension of \jr to the case of SCV is to impose a \jr-type condition across all sub-committees.\footnote{Imposing representation requirements across all sub-committees implicitly assumes that the selections of all sub-committees are of comparable significance to the voters.} The definition of SW-JR is identical to the definition of \jr for the committee voting setting. 

\begin{definition}[Span-wise JR (SW-JR)]
	An SCV outcome $W$ satisfies \emph{Span-wise JR (SW-JR)} if $\forall X\subseteq N$: 
	\[|X|\geq \frac{n}{k} \text{ and } |(\cap_{i\in X}A_i)|\geq 1 \\ \implies (|W \cap (\cup_{i\in X}A_i)|\geq 1)\]
	\end{definition}
	
	 This approach to representation leads to a `global' \jr axiom which aims to represent large, cohesive groups of voters (i.e. $X\subseteq N: \, |X|\ge n/k$ and $\cap_{i\in X} A_i\neq \emptyset$) in some sub-committee, but not necessarily a sub-committee where they are cohesive. 
	 
	Note that both SW-JR and IW-JR concern representation that are not at the level of single individual but at the level of large enough cohesive groups.\footnote{In a related paper, \citet{CFS17a} proposed fairness concepts for Public Decision Making that is equivalent to SCV in which $k_i=1$ for each $i$. They considered different fairness notions that are based on proportional or envy-free allocations. The concepts involve viewing agents independently and are different from proportional representation concerns. When the number of sub-committees is less than the number of voters, the concepts they consider are trivially satisfied.}
	
	Next we show that a SW-JR committee may not exist and is NP-hard to compute. 

%

%
%

\section{(Non)-existence and complexity of SW-JR committees}

We show that a committee satisfying SW-JR may not exist under either of the two restriction (1)  there are exactly two candidate subsets, and (2) $k_i=1$ for all $i\in \{1,\ldots, \ell\}$.

%
%
%

\begin{proposition}
A committee satisfying SW-JR may not exist even if there are exactly two candidate subsets and $k_i=1$ for $i=1,2$. 
\end{proposition}

\begin{proof}
Consider an SCV instance where $|N|=n=2$, $C=C_1\cup C_2$ with $C_1=\{a_1, a_2\}$ and $C_2=\{b_1, b_2\}$, and $k_1=k_2=1$. Note that $n/k=1$. 

If the approval ballots are $A_1=\{a_1\}$ and $A_2=\{a_2\}$, then there is no SCV outcome $W$ (i.e. a committee) which satisfies SW-JR. This can be immediately observed since SW-JR requires both voters to be represented, however the quota $k_1=1$ prevents this from being possible.
\end{proof}

The reader may note that above proof utilises an example where the voters have ballots which do not approve of any voter in some candidate subset (i.e. $A_i\cap C_2=\emptyset$). This feature is not required to show the non-existence of an SW-JR committee however, it greatly simplifies the example.

		Above we proved that a committee satisfying SW-JR may not exist. One could still aim to find such a committee whenever it exists. Next we prove that the problem of checking whether a  SW-JR committee exists or not is NP-complete.
	
			\begin{proposition}\label{prop:SWJR-NPC}
				Checking whether an SW-JR committee exists or not is NP-complete.
				\end{proposition}

			To show that checking whether a SW-JR committee exists is NP-complete we will reduce a given instance of the \textsc{Set Cover} problem, a known NP-complete problem ~\cite{GaJo79a}, to an SCV instance -  such that an SW-JR committee exists if and only if the \textsc{Set Cover} instance has a yes answer.
				
				Below is a statement of the \textsc{Set Cover} problem.

					           \pbDef{\textsc{ Set Cover}}
					  {Ground set $X$ of elements, a collection $\mathcal{L}=\{S_1,\ldots, S_{t}\}$ of subsets of $X$ such that $X=\cup_{j=1}^t S_j$ and an integer $k'$.}
					           			 {Does there exist a $H\subset \mathcal{L}$ such that $|H|\leq k'$ and  $X=\bigcup_{S_j\in H}S_j$}

							 
 To build intuition for the formal proof, which follows, we provide an overview of the reduction.
 
  The \textsc{Set Cover} problem involves answering whether or not there exists a collection of at most $k'$ subsets $H \subseteq \mathcal{L}$ which cover another set $X$. This problem can be embedded into an SCV instance by letting the set $X$ represent the set of voters and considering a candidate subset $C_2$ such that each element denotes an element of $\mathcal{L}$; that is,
  $$C_2=\{s_1, \ldots, s_t\}.$$
We then let each element of $C_2$, say $s_j$, be approved by voters $i\in N=X$ if and only if $i\in S_j$. By appropriately defining quota values and voter approval ballots on the remaining candidate subset $C_1$ it is shown that an SW-JR committee exists if and only if every voter is represented via a candidate in $C_2$ -- this of course possible if and only if we have a yes-instance of the \textsc{Set Cover} problem


  
  

A formal proof is presented below.

	\begin{proof}
		We reduce a given \textsc{Set Cover} instance $(X, \mathcal{L}, k')$ to an SCV instance as follows: Let 
	\begin{align*}
	N&=X=\{1, \ldots, n\}\\
	C&=C_1\cup C_2\\
	&=\{a_1, \ldots, a_n\}\cup\{s_1, \ldots, s_t\} 
	\end{align*} 
	denote the set of voters, candidate set and partition into two candidate subsets. Let voter approval ballots be
	$$A_i=\{s_j | i\in S_j\},$$
	and $k_1=n, k_2=k'$.  Without loss of generality we may assume that $t\ge k'$, also note that $k=n+k'$ and so $n/k<1$. \\
	
	Since every voter has a non-empty approval ballot (i.e. $N=X=\cup_{j=1}^t S_j$) and $n/k<1$, a committee satisfies SW-JR if and only if every voter is represented. 
	
	If $(X, \mathcal{L}, k')$ is a yes-instance of the \textsc{Set Cover} problem then these exists a subset $H$ with $|H|=k'$ such that $N=X=\bigcup_{S_j\in H}S_j$ -- it follows that the committee 
	$$W=C_1\cup\{s_i \,| \,S_i\in H\},$$
	is a solution to the SCV problem and satisfies SW-JR. Conversely, if $W$ is an SW-JR committee then the set
	$$H=\{S_i\, |\, s_i\in W\},$$
	provides a yes-instance to the  \textsc{Set Cover} problem.

	\end{proof}

\section{Weak-SW-JR}

	In the previous section we showed that a committee satisfying SW-JR need not exist, and checking whether it does is NP-complete. Naturally and in the pursuit of a computationally tractable representation axiom we weaken the concept of SW-JR. In this section we present a weak version of SW-JR, appropriately referred to as weak-SW-JR. An SCV outcome satisfying weak-SW-JR is guaranteed to exist and is attainable via a polynomial-time algorithm.

	\begin{definition}[Weak-SW-JR] 
	An SCV outcome $W$ satisfies \emph{weak-Span-wise Justified Representation (weak-SW-JR) }if
	\begin{align*}
	\forall X&\subseteq N: |X|\ge \frac{n}{k} \text{ and } \Big|\Big(\bigcap_{i\in X} A_i\Big)\cap C_j\Big|\ge 1 \quad \text{for all } j\\
	& \implies  \Big|W\cap\Big(\bigcup_{i\in X} A_i\Big)\Big|\ge 1
	\end{align*}
		\end{definition}
		
 Informally speaking, the weak-SW-JR axiom captures the idea that if a ``large", cohesive set of voters unanimously support at least one candidate in each candidate subset then they require representation in some sub-committee. 


 First observe that weak-SW-JR is indeed a (strict) weakening of the SW-JR concept. 

	\begin{proposition}
	SW-JR implies Weak-SW-JR. But weak-SW-JR does not imply SW-JR.
	\end{proposition}

	\begin{proof}
	We prove the proposition via the contrapositive, suppose weak-SW-JR does not hold. Then there exists a set $X\subseteq N$ such that $|X|\ge n/k$ with $A_i\cap W=\emptyset$ for all $i\in X$ and $(\bigcap_{i\in X} A_i )\cap C_j\neq \emptyset\,\,$ for all $j=1,\ldots, \ell.$
	But 
	$$\emptyset \neq \Bigg(\bigcap_{i\in X} A_i \Bigg)\cap C_j\subseteq \Bigg(\bigcap_{i\in X} A_i \Bigg)$$
	and so SW-JR does not hold.

 The second claim can be easily observed from the definition and simple counter examples can be constructed (for an example see within the proof of Proposition~\ref{prop-relations} in the supplement material).
	\end{proof}

The next proposition states that for a given SCV instance there may be three distinct committees $W$ satisfying, respectively, weak-SW-JR but not IW-JR, IW-JR but not weak-SW-JR, and both weak-SW-JR and IW-JR simultaneously. That is, the weak-SW-JR and IW-JR representation axioms are distinct but are not mutually exclusive (the proof can be found in the supplement material). 

	\begin{proposition}\label{prop-relations}
	Weak-SW-JR does not imply IW-JR and IW-JR does not imply weak-SW-JR. Also weak-SW-JR and IW-JR are not mutually exclusive concepts.
	\end{proposition}
		\begin{proof}
		We provide an example of an SCV instance which admits three distinct SCV outcomes which satisfy
		\begin{enumerate}
		\item weak-SW-JR but not IW-JR
		\item IW-JR but not weak-SW-JR
		\item both weak-SW-JR and IW-JR.
		\end{enumerate}

		\noindent Consider the SCV with $|N|=12$,
		$$C=C_1\cup C_2,$$
		 where $C_1=\{c_1, \ldots, c_7\}$, $C_2=\{a, b, c\}$, $k_1=1$ and $k_2=2$. Let the approval ballots of each voter be as follows
		\begin{align*}
		A_i&=\begin{cases}
		\{c_i, a\} &\text{if $i\in \{1, \,2,\,3,\,4,\,5,\, 6\}$}\\
		\{c_7, b\} &\text{if $i\in\{7,\, 8,\,9,\,10\}$}\\
		\{c_8, c\} &\text{if $i\in\{11, 12\}$}
		\end{cases}
		\end{align*}

		First observe that for an SCV outcome $W$ to satisfy weak-SW-JR the only set of voters which must be represented is $X=\{7,8,9,10\}$ since they are a the only group of size greater or equal to $n/k=4$ who unanimously support a voter in each of the candidate subsets i.e. $c_7\in C_1$ and $b\in C_2$ are approved by every voter $i\in X$. 

		 Whilst, for an SCV outcome to satisfy IW-JR it is required that the group $X'=\{1, 2, \ldots, 6\}$ are represented in $C_2$ by candidate $a\in C_2$. This is because $X'$ is the only group of size greater or equal to $n/k_2=6$ who unanimously support a candidate in $C_2$.  
 
	Thus it follows immediately that there exists three distinct SCV outcomes which satisfy the three properties stated at the beginning of this proof, namely;
		$$W^{\textrm{weak-SW-JR}}=\{c_1, b, c\}$$
		satisfies weak-SW-JR but not IW-JR,
		$$W^{\textrm{IW-JR}}=\{c_1, a,c\}$$
		satisfies IW-JR but not weak-SW-JR and,
		$$W^{\textrm{weak-SW-JR and IW-JR}}=\{c_7, a, c\}.$$
	        satisfies both IW-JR and weak-SW-JR.
		\end{proof}

\section{An Algorithm for weak-SW-JR and IW-JR}

	The previous sections have introduced two appealing representation axioms, weak-SW-JR and IW-JR, which capture distinct notions of representation or fairness. The former axiom considers the structure of approvals across all candidate subsets, whilst the latter axiom considers each candidate subset as an independent event. This section will combine these two axioms and consider SCV outcomes $W$ which satisfy both weak-SW-JR and IW-JR. We prove that a committee satisfying both weak-SW-JR and IW-JR is guaranteed to exists for every SCV setting and such a committee can be computed in polynomial-time.
	
	We begin by presenting the following intermediate fact before providing a constructive existence proof of a committee which satisfies weak-SW-JR and IW-JR.


	\begin{lemma}\label{algo-lemma}
	Let $\{k_j\}_{j=1}^\ell$ be a sequence of positive numbers and $\{k_j'\}_{j=1}^\ell$ be a sequence of non-negative numbers. Then,
	$$\sum_{j=1}^\ell \frac{k_j'}{\sum_{i=1}^\ell k_i} \le \max_{j\in [\ell]} \Big\{\frac{k_j'}{k_j}\Big\}.$$
	\end{lemma}

	\begin{proof}
	Let $M=\max_{j\in [\ell]} \Big\{\frac{k_j'}{k_j}\Big\}$ and let $w_j=\frac{k_j}{\sum_i k_i}$, then
	
	
	\begin{align*}
	M&= \sum_j w_j M\ge \sum_j w_j \frac{k_j'}{k_j}
	=\sum_j \frac{k_j'}{\sum_i k_i}, \quad \text{as required.}
	\end{align*}
	\end{proof}

							\begin{algorithm}[h!]
								  \caption{Algorithm that returns a committee which satisfies both weak-SW-JR and IW-JR. }
								  \label{algo:Weak SW-JR-and-IW-JR}

								\begin{algorithmic}
									\REQUIRE  SCV instance $(N,C,\pi,q,\pref)$
									\ENSURE Committee $W$ that satisfies both weak-SW-JR and IW-JR.
								\end{algorithmic}
								\begin{algorithmic}[1]
									\STATE $W_j=\emptyset$ for $j=1, \ldots, \ell$, $W=\bigcup_{j} W_j$
									\STATE For each candidate subset $j=1, \ldots, \ell$, allocate the candidates with support $\ge n/k_j$ from highest to lowest in support, removing the support of voters who are already represented in $W_j=W\cap C_j$. 
									\STATE From the remaining positions, consider unelected candidates who have support $\ge n/k$ among the unrepresented voters. Allocate these from highest to lowest in support removing the support of voters who are already represented in $W$.
									\STATE If there are remaining positions allocate in any way.
%
%
%
%
								\end{algorithmic}
							\end{algorithm}

	\begin{proposition}
	A committee which satisfies both weak-SW-JR and IW-JR always exists and can be attained via Algorithm~\ref{algo:Weak SW-JR-and-IW-JR}.
	\end{proposition}

	\begin{proof}

Consider Algorithm~\ref{algo:Weak SW-JR-and-IW-JR}. We argue that the committee $W$ returned by Algorithm~\ref{algo:Weak SW-JR-and-IW-JR} satisfies both weak-SW-JR and IW-JR.

\bigskip

	First we show that IW-JR is satisfied. Suppose not, then during step 2 for some $j$ we allocated $k_j$ winning spots but failed to represent a group, say $X\subseteq N$, of size at least $n/k_j$ who unanimously supports some candidate(s) in $C_j$. However, at each stage at least $n/k_j$ additional voters are represented and so it must be the case that at least $k_j n/k_j=n$ voters were represented in $W$. That is, all voters have been represented which contradicts the existence of the set $X$. Thus IW-JR is always satisfied.



 Now we show that weak-SW-JR is also satisfied. Suppose that after step 2, $|W_j|=k_j'<k_j$ for all $j$ -- if this were not the case then weak-SW-JR is trivially satisfied since all voters would then be represented in $W$. The proof now reduces to showing that there are enough `places' left in $W$ after allocating the $\sum_j k_j'$ places in Step 2.
	
In the `worst case' every allocation for each $j$ represents the same subset of voters -- in this case $W$ represents $\ge R:= \max_j \{k_j' \frac{n}{k_j}\}$ voters with $\sum_j k_j'$ elected candidates. But then there are
	$$\le \frac{n-R}{n/k} =\frac{n-n\max_j\{\frac{k_j'}{k_j}\}}{n/k}=k\Big(1-\max_j\big\{\frac{k_j'}{k_j}\big\}\Big)$$
	possible mutually exclusive groups of size $\ge n/k$ which are unrepresented in $W$. In the worst case, each of these groups would unanimously support a different candidate in every candidate subset and so correspond to a problem set with respect to weak-SW-JR. Recall that we have $k-\sum_j k_j'$ winning spots left and so suffices to show that
	\begin{align}\label{applylem1}
	k-\sum_j k_j' \ge k\Big(1-\max_j\big\{\frac{k_j'}{k_j}\big\}\Big).
	\end{align}
	Note that we use the property that if problem set $X$ exists they must unanimously support some candidate in every sub-committee and so we can ignore the quota issues. Finally, (\ref{applylem1}) follows immediately from dividing by $k$ and applying Lemma~\ref{algo-lemma}.
	\end{proof}
	


		\section{Testing Representation}

 Testing for SW-JR and IW-JR are easy given the polynomial-time testing of JR. Testing for weak-SW-JR is more involved, we conjecture the complexity is coNP-complete. Before we proceed, we outline the standard approval-based voting setting and an algorithm identified by \citet{ABC+16a} to test JR.
\vspace{1em}


	\noindent \textbf{Polynomial-time algorithm to verify JR}
	
	The standard setting of approval-based voting (AV) is a special case of SCV. In particular, an AV instance is a tuple $(N, C, k, A)$ where $N$ is a set of voters, $C$ is the set of candidates, $k$ is a positive integer and $A$ is an approval ballot profile.  An AV outcome (or committee) is a subset $W\subseteq C$ such that $|W|=k$. Note, that this is simply a special case of SCV when $\ell=1$, $C_1=C$ and $q(C_1)=k$.

 The algorithm proposed by \citet{ABC+16a} to test JR is as follows: given an Approval Voting instance $(N, C, k, A)$ and outcome $W$, for each candidate $c\in C$ compute 
	$$s(c)=\big| \{i\in N\, : \, c\in A_i, \, A_i\cap W=\emptyset\}\big|.$$
	The set $W$ fails to provide JR for $(A, k)$ if and only if there exists a candidate $c$ with $s(c)\ge n/k$.\\

With minor modifications, the above algorithm provides a polynomial-time algorithm to test whether an SCV outcome satisfies SW-JR and IW-JR. 

	\begin{proposition}
	It can be checked in polynomial time whether  a given committee satisfies SW-JR or not.
	\end{proposition}

	\begin{proof}
	Same as the polynomial-time algorithm for testing JR.
	\end{proof}
	
	\begin{proposition}
	It can be checked in polynomial time whether  a given committee satisfies IW-JR or not.
	\end{proposition}

	\begin{proof}
	First note that an SCV outcome $W$ satisfies IW-JR if and only if it satisfies JR for every candidate subset $C_i$ when approvals ballots are restricted to $C_i$. That is, $W$ satisfies IW-JR if and only if the Approval Voting instance $(N, C_i, k_i, A)$ satisfies JR for all $i\in [\ell]$. 
	
	It follows immediately that applying the polynomial-time algorithm to verify that JR is satisfied in each of these $\ell$ Approval Voting instances is also a polynomial-time algorithm. 
		\end{proof}

\section{Generalizing PAV to SCV} 

In the setting of approval-based multi-winner voting, the Proportional Approval Voting (PAV) rule has been extensively studied and shown to satisfy many desirable representation properties. It has been shown in ~\citep{ABC+16a} that PAV committees satisfy JR, though computing a PAV committee is  NP-hard.\footnote{In fact PAV is viewed as one of the most compelling rules for approval-based committee voting because it satisfies EJR a property stronger than JR~\citep{ABC+16a}.}

Under PAV in the standard approval voting setting (AV), each voter who has $j$ of their approved candidates in the committee $W$ is assumed to derive utility of $r(j):=\sum_{p=1}^j 1/p$ if $j>0$ and zero otherwise. The total utility of a committee $W$ is then defined as 
$$\text{PAV}(W)=\sum_{i\in N} r(|W\cap A_i|),$$
this is known as the PAV-score. The PAV rule outputs the committee $W$ of size $k$ which maximizes the PAV-score among all committees of size $k$.



In this section we consider generalizing the PAV rule to the SCV setting. This leads to two distinct PAV rules for the SCV setting which are both natural generalizations. 



 \textbf{Span-wise PAV (SW-PAV)} is a generalization of the PAV rule which assumes voters gain utility solely from the number of their approved candidates in $W$. Thus, each voter $i$ derives utility $r(|W\cap A_i|)$ and the SW-PAV score of a committee $W$ is 
$$\text{SW-PAV}(W)=\sum_{i\in N} r(|W\cap A_i|).$$

\textbf{Intra-wise PAV (IW-PAV)} is a generalization of the PAV rule which assumes voters gain utility from both the number of their approved candidates in $W$ and also the diversity across sub-committees. Thus each voter $i$ derives utility 
$$\rho(|W\cap A_i|):= \sum_{j \in [\ell]} r(|W\cap A_i\cap C_j|),$$
and the IW-PAV score of a committee $W$ is 
$$\text{IW-PAV}(W)=\sum_{i\in N} \rho(|W\cap A_i|)=\sum_{j\in [\ell]} \sum_{i\in N} r(|W\cap A_i\cap C_j|).$$


In both generalizations, a SW-PAV (IW-PAV) committee is defined to be a committee $W$ satisfying the SCV quota conditions and maximizing the SW-PAV (IW-PAV) score.

To illustrate the distinction between SW-PAV and IW-PAV the following example is provided.

\begin{example}
Consider a voter $i$ with approved and elected candidates $\{a, b,c\}$ such that $\{a, b\}\in C_1$ and $c\in C_3$. Then voter $i$'s contribution to the SW-PAV score is
$$r(|W\cap A_i|)=1+\frac{1}{2}+\frac{1}{3}=1\frac{5}{6}.$$
Whilst her contribution to the IW-PAV score is
$$\rho(|W\cap A_i|)=r(2)+r(1)=1+\frac{1}{2}+1=2\frac{1}{2}.$$
\end{example}

\paragraph{\textbf{Generalized Justified Representation under SW-PAV and IW-PAV}}

We consider the representation properties of SCV committee outcomes from the IW-PAV and SW-PAV rules. We show that IW-PAV satisfies IW-JR and SW-PAV satisfies weak-SW-JR, whilst neither satisfies both. In addition, we show that both rules can fail to output an SW-JR committee when such a committee exists.


\begin{proposition}
IW-PAV satisfies IW-JR.
\end{proposition}

\begin{proof}
Any IW-PAV maximizing committee $W$ must be such that for all $i$ the set $W\cap C_i$ is a PAV maximizing committee in the standard AV setting. Hence JR must be satisfied in each $C_i$, as shown in  ~\citep{ABC+16a}, thus IW-JR is satisfied.
\end{proof}

We now show that SW-PAV satisfies weak-SW-JR, however first we introduce some notation and a lemma.

Let $W$ be an SCV committee and let $c\in W$, define the marginal contribution of $c$ as
$$MC(c, W)=\text{SW-PAV}(W)-\text{SW-PAV}(W-\{c\}).$$


The following lemma was explicitly presented by \citet{AzHu17a} for the standard approval voting setting and was implicitly used in \citep{ABC+16a}. The lemma applies to the SCV setting via an identical argument. We omit the proof and provide a reference.


\begin{lemma}\label{SPAVlemma0}\emph{[\citet{AzHu17a}]}
For any committee $W$ such that $|W|=k$, there exists at least one $c\in W$ such that 
$$MC(c,W)\le |\{i\in N: A_i\cap W\neq \emptyset\}|/k\le n/k.$$
\end{lemma}


\begin{proposition}
SW-PAV satisfies weak-SW-JR.
\end{proposition} 

\begin{proof}
Suppose for the purpose of a contradiction that $W^*$ is a SW-PAV committee which does not satisfy weak-SW-JR. That is, there exists a group of unrepresented voters $X$ with $|X|\ge n/k$ such that $A_i\cap W^*=\emptyset$ for all $i\in X$ and for all $j\in [\ell]$ $\bigcap_{i\in X} A_i\cap C_j\neq \emptyset$.

First, note that there must exist a candidate $c\in W^*$ such that $MC(c,W^*)<n/k$. Suppose otherwise, then by Lemma~\ref{SPAVlemma0} it must be that all $n$ voters are represented which contradicts the existence of the set $X$.



Now suppose that $c\in C_j$ for some $j\in [\ell]$ and let $c'\in C_j$ such that $c'\in \bigcap_{i\in X} A_i\cap C_j$. Then it is clear that adding $c'$ to the committee $W^*-\{c\}$ increases the SW-PAV score by at least $n/k$ and so 
$$\text{SW-PAV}\Big((W^*-\{c\})\cup\{c'\}\Big)>\text{SW-PAV}(W^*),$$
which contradicts $W^*$ being a SW-PAV committee. Thus, a SW-PAV committee must satisfy weak-SW-JR.  
\end{proof}

The following proposition shows that both SW-PAV and IW-PAV can fail to produce a SW-JR committee when such a committee exists. The proof illustrates a trade-off between maximizing voter utility and pursuing the representation axiom of SW-JR.


\begin{proposition}
Both SW-PAV and IW-PAV can fail SW-JR when a SW-JR committee exists. 
\end{proposition}

\begin{proof}
Consider the following counter-example. Let $N=\{1, \ldots, 12\}$,
$$C=C_1\cup C_2=\{a,b,c\} \cup \{a_1, \ldots, a_5, b_1, \ldots, b_5, c_1\},$$
 and let the approval ballots be 
\begin{align*}
\qquad \qquad \quad A_i&=\{a, a_i\} ~~&&\text{for all } i\in \{1, \ldots, 5\}\\
 \qquad \qquad \quad  A_i&=\{b, b_{i-5}\} ~~&&\text{for all } i\in \{ 6, \ldots, 10\}\\
\qquad  \qquad \quad  A_i&=\{c, c_1\} ~~&&\text{for all } i\in \{11, 12\},
\end{align*}
with quotas $k_1=1$ and $k_2=2$.


For SW-JR to be satisfied we require that at least one voter from each of the groups $X=\{1, 2, 3, 4, 5\}$ and $X'=\{6, 7, 8,9,10\}$ are represented in $W$. Thus, the only committees satisfying SW-JR are of the following form:
$$W_1=\{a, b_i, b_j\} \quad \text{or} \quad W_2=\{b,a_i, a_j\} \quad \text{with $i\neq j$}.$$
Any committee with the form above has an SW-PAV score of 7 and an IW-PAV score of 7. However the committee $W^*=\{a, c_1, b_1\}$ maximizes the SW-PAV and IW-PAV scores, with both equal to 8, and does not satisfy SW-JR. 
\end{proof}

We now present two propositions showing that SW-PAV can fail to produce a committee satisfying IW-JR, and IW-PAV can fail weak-SW-JR. 

The proof of the following proposition highlights the conflict between SW-PAV, which does not incentivize diversity, and IW-JR, which may demand representation of a group already represented in another sub-committee. 


\begin{proposition}
SW-PAV can fail to produce a committee satisfying IW-JR.
\end{proposition}

\begin{proof}
Let $N=\{1, \ldots, 12\}$, 
$$C=\{a, b\}\cup\{a', b'\}\cup\{a'', b'', c''\},$$
 and let approval ballots be
\begin{align*}
A_i=\begin{cases}
&\{a, a', a''\} ~~\text{if $i\in \{1, \ldots, \, 6\}$}\\
&\{a, a', b''\} ~~\text{if $i\in \{7,\,8\}$}\\
&\{a, a', c''\} ~~\text{if $i\in \{9\}$}\\
&\{b, b', b''\} ~~\text{if $i\in \{10\}$}\\
&\{b, b', c''\} ~~\text{if $i\in \{11, \, 12\}$}\\
\end{cases}
\end{align*}
with quotas $k_1=k_2=1$ and $k_3=2$. 

The only candidate required to be in a candidate subset for IW-JR is candidate $a''\in C_3:= \{a'', b'', c''\}$ since $12/2=6$. However, direct computation shows that a maximal SW-PAV score is $15\frac{5}{6}$ attained from $W^*=\{a, a', c''\}$.
\end{proof}

The following proposition shows that IW-PAV can fail weak-SW-JR. The proof highlights the conflict between IW-PAV, which incentives diversity, and weak-SW-JR, which may require a smaller unrepresented group of voters to be represented. 


\begin{proposition}
IW-PAV can fail to produce a committee satisfying weak-SW-JR
\end{proposition}

\begin{proof}
Let $N=\{1, \ldots, 12\}$, $C=\{a, b, c\}\cup \{a', b', c'\},$ let approval ballots be 
\begin{align*}
A_i&=\{a, a'\} \quad \text{for all } i\in \{ 1, \ldots, 5\},\\
A_i&=\{b, b'\} \quad\text{for all }~ i\in \{6, \ldots, 9\},\\
A_i&=\{c, c'\} \quad\text{for all }~ i\in \{10, 11,12\},
\end{align*}
with quotas $k_1=k_2=2$ . Note $n/k=3$.

For a weak-SW-JR committee every voter must be represented. However, the IW-PAV committee is $W^*=\{a, b, a', b'\}$ which gives an IW-PAV score of $9\times (1+1)=18$ and does not satisfy weak-SW-JR.
\end{proof}





%
%
%
%

\section{Discussion}

In this paper we formalized a general social choice model called sub-committee voting. 
We focussed on natural generalization of JR from the approval-based committee voting setting to the approval-based SCV setting. 
Some of the results are summarized in Table~\ref{table:summary}.
It will be interesting to consider generalizations of stronger versions of justified representation such as PJR and EJR. For example, IW-JR can straightforwardly be strengthened to IW-PJR or IW-EJR. 

It will be interesting to consider more general preferences that need not be approval-based. 
Several research questions that have been intensely studied in subdomains of SCV apply as well to SCV. For example, it will be interesting to extend axioms and rules for single-winner or multi-winner voting to that of SCV.

						            			\begin{table}[h!]
						            				\centering
						            			\label{tab:compare}
										\scalebox{1}{
						            			\begin{tabular}{lccc}
						            			\toprule
		&Representative &Complexity\\
		&committee &of\\
		&exists&computing\\
		\midrule								SW-JR&No&NP-c\\
										IW-JR&Yes&in P\\
										weak-SW-JR&Yes&in P\\
										IW-JR \& weak-SW-JR&Yes&in P\\
		%
						            			\bottomrule
						            			\end{tabular}
										}
										\caption{Properties of justified representation concepts for sub-committee voting.}
			            			\label{table:summary}
						            			\end{table}


In any combinatorial setting, one can view the voting process as either simultaneous voting or sequential voting~\citep{BaLa16a,LaXi15a,FZC17a,FZC17b}. 
We formalized SCV as a static model in which $\ell$ sub-committees are to be selected simultaneously. The representation notions that we formalized can also be considered if voting over each sub-committee is conducted sequentially over time. The axioms that we consider such as SW-JR apply as well to understand the quality of an outcome in these online or sequential settings. 
	\section*{Acknowledgments}
	Haris Aziz is supported by a Julius Career Award.
	Barton Lee is supported by a Scientia PhD fellowship.

\appendix

%

		\end{document}